\documentclass[11pt]{article}
\usepackage{a4wide}
\usepackage{fullpage}
\usepackage{my-preamble}
\usepackage{times}

\newcommand{\etalchar}[1]{$^{#1}$}
\title{Near-Optimal Expanding Generating Sets for Solvable Permutation Groups}

\author{V.~Arvind\footnotemark[2] 
\and Partha Mukhopadhyay\footnotemark[1] \and 
Prajakta Nimbhorkar 
\thanks{Chennai Mathematical Institute, Siruseri, 
India. Emails: {\tt \{partham,prajakta\}@cmi.ac.in}}
\and Yadu Vasudev \thanks{The Institute of Mathematical Sciences, Chennai, India.Emails: {\tt \{arvind,yadu\}@imsc.res.in}}}
\date{}

\begin{document}
\maketitle{}
\begin{abstract}
  Let $G =\langle S\rangle$ be a solvable permutation group of the
  symmetric group $S_n$ given as input by the generating set $S$.  We give
  a deterministic polynomial-time algorithm that computes an
  \emph{expanding generating set} of size $\tilde{O}(n^2)$ for $G$. More
  precisely, the algorithm computes a subset $T\subset G$ of size
  $\tilde{O}(n^2)(1/\lambda)^{O(1)}$ such that the undirected Cayley graph
  $Cay(G,T)$ is a $\lambda$-spectral expander (the $\tilde{O}$ notation
  suppresses $\log ^{O(1)}n$ factors).  As a byproduct of our proof, we get
  a new explicit construction of $\varepsilon$-bias spaces of size
  $\tilde{O}(n\poly(\log d))(\frac{1}{\varepsilon})^{O(1)}$ for the groups
  $\Z_d^n$. The earlier known size bound was
  $O((d+n/\varepsilon^2))^{11/2}$ given by \cite{AMN98}. 

\end{abstract}
\newpage
\section{Introduction}
Expander graphs are of great interest and importance in theoretical
computer science, especially in the study of randomness in
computation; the monograph by Hoory, Linial, and Wigderson
\cite{HLW06} is an excellent reference. A central problem is the
explicit construction of expander graph families
\cite{HLW06,LPS88}. By explicit it is meant that the family of graphs
has efficient deterministic constructions, where the notion of
efficiency depends upon the application at hand,
e.g. \cite{Re08}. Explicit constructions with the best known and near
optimal expansion and degree parameters (the so-called Ramanujan
graphs) are Cayley expander families \cite{LPS88}.

Alon and Roichman, in \cite{AR94}, show that every finite group has a
logarithmic size expanding generating set using the probabilistic
method. For any finite group $G$ and $\lambda>0$, they show that with
high probability a random multiset $S$ of size 
$O(\frac{1}{\lambda^2}\log |G|)$ picked
uniformly at random from $G$ is a $\lambda$-spectral
expander. Algorithmically, if $G$ is given as input by its
multiplication table then there is a randomized \emph{Las Vegas} algorithm
for computing $S$: pick the multiset $S$ of $O(\frac{1}{\lambda^2}\log |G|)$ many
elements from $G$ uniformly and independently at random 
and check in deterministic time $|G|^{O(1)}$ that
$Cay(G,T)$ is a $\lambda$-spectral expander.

Wigderson and Xiao gave a derandomization of this algorithm in
\cite{WigdersonX08}(also see \cite{ArvindMN11} for a new combinatorial
proof of \cite{WigdersonX08}). Given $\lambda>0$ and a finite group
$G$ by a multiplication table, they show that in deterministic time
$|G|^{O(1)}$ a multiset $S$ of size 
$O(\frac{1}{\lambda^2}\log|G|)$ can be computed such
that $Cay(G,T)$ is a $\lambda$-spectral expander.

\subsection*{This paper}

Suppose the finite group $G$ is a subgroup of the symmetric group
$S_n$ 
and $G$ is given as input by a
\emph{generating set} $S$, and not explicitly by a multiplication
table. The question we address is whether we can compute an $O(\log
|G|)$ size expanding generating set for $G$ in deterministic polynomial
time. 

Notice that if we can randomly (or nearly randomly) sample from
the group $G$ in polynomial time, then the Alon-Roichman theorem
implies that an $O(\frac{1}{\lambda^2}\log |G|)$ size sample will be an 
$(1-\lambda)$-expanding
generating set with high probability. Moreover, it is possible to sample 
efficiently and near-uniformly from any black-box group given by a set of 
generators \cite{Bab91}.   


This problem can be seen as a generalization of the construction of
small bias spaces in say, $\F_2^n$ \cite{AGHP92}. It is easily proved
(see e.g.\ \cite{HLW06}), using some character theory of finite
abelian groups, that $\varepsilon$-bias spaces are precisely expanding
generating sets for $\F_2^n$ (and this holds for any finite abelian
group). Interestingly, the best known explicit construction of
$\varepsilon$-bias spaces is of size either $O(n^2/\varepsilon^2)$ 
\cite{AGHP92} or 
$O(n/\varepsilon^3)$ \cite{ABNNR92}, whereas the
Alon-Roichman theorem guarantees the existence of $\varepsilon$-bias
spaces of size $O(n/\varepsilon^2)$.

Subsequently, Azar, Motwani and Naor \cite{AMN98} gave a construction
of $\varepsilon$-bias spaces for finite abelian groups of the form
$\Z_d^n$ using Linnik's theorem and Weil's character sum bounds.
The size of the $\varepsilon$-bias space they give is 
$O((d+n^2/\varepsilon^2)^C)$
where the constant $C$ comes from Linnik's theorem and the current best 
known bound for $C$ is $11/2$.

Let $G$ be a finite group, and let $S=\langle
g_1,g_2,\ldots,g_k\rangle$ be a \emph{generating} set for $G$. The
\emph{undirected Cayley graph} $\Cay(G,S\cup S^{-1})$ is an undirected
multigraph with vertex set $G$ and edges of the form $\{x,xg_i\}$
for each $x\in G$ and $g_i\in S$. Since $S$ is a generating set for
$G$, $\Cay(G,S\cup S^{-1})$ is a connected regular multigraph.

In this paper we prove a more general result. Given any solvable
subgroup $G$ of $S_n$ (where $G$ is given by a generating set) and 
$\lambda >0$, we
construct an expanding generating set $T$ for $G$ of size 
$\tilde{O}(n^2)(\frac{1}{\lambda})^{O(1)}$ such that $\Cay(G,T)$
is a $\lambda$-spectral expander. 
We also note that, for a \emph{general}
permutation group $G\leq S_n$ given by a generating set, we can compute
(in deterministic polynomial time) an $O(n^c)(\frac{1}{\lambda})^{O(1)}$ 
size generating set $T$
such that $\Cay(G,T)$ is a $\lambda$-spectral expander. Here $c$ is a large
absolute constant.

Now we explain the main ingredients of our
expanding generating set construction for solvable groups:

\begin{enumerate}
  \item\label{it:normal-quotient} Let $G$ be a finite group and $N$ be a
    normal subgroup of $G$.  Given expanding generating sets $S_1$ and
    $S_2$ for $N$ and $G/N$ respectively such that the corresponding Cayley
    graphs are $\lambda$-spectral expanders, we give a simple
    polynomial-time algorithm to construct an expanding generating set $S$
    for $G$ such that $\Cay(G,S)$ is also $\lambda$-spectral expander.  
    Moreover,
    $|S|$ is bounded by a constant factor of $|S_1|+|S_2|$. 
  \item\label{it:solvable} We compute the derived series for the given
    solvable group $G\le S_n$ in polynomial time using a standard algorithm
    \cite{Luk93}. This series is of $O(\log n)$ length due to Dixon's
    theorem. Let the derived series for $G$ be $G=G_0\rhd G_1\rhd \cdots
    \rhd G_k=\{1\}$.  Assuming that we already have an expanding generating
    set for each quotient group $G_i/G_{i+1}$ (which is abelian) of size
    $\tilde{O}(n^2)$, we apply the previous step repeatedly to obtain an
    expanding generating set for $G$ of size $\tilde{O}(n^2)$. We can do
    this because the derived series is a normal series. 
  \item\label{it:abelian-group} Finally, we consider the abelian quotient
    groups $G_i/G_{i+1}$ and give a polynomial time algorithm to construct
    expanding generating sets of size $\tilde{O}(n^2)$ for them. This
    construction applies a series decomposition of abelian groups as well
    as makes use of the Ajtai et al construction of expanding generating
    sets for $\Z_t$ \cite{AIKPS90}. 
\end{enumerate}

We describe the steps \ref{it:normal-quotient}, \ref{it:solvable} and
\ref{it:abelian-group} in Sections \ref{sec:normal-quotient},
\ref{sec:solvable} and \ref{sec:abelian-group} respectively.  As a simple
application of our main result, we give a new explicit construction of
$\varepsilon$-bias spaces for the groups $\Z_d^n$ which we explain in
Section \ref{sec:discussion}. The size of our $\varepsilon$-bias spaces are
$O(n\poly(\log n,\log d)) (\frac{1}{\varepsilon})^{O(1)}$. To the best of our
knowledge, the known construction of $\varepsilon$-bias space for $\Z_d^n$
gives a size bound of 
$O((d+n/\varepsilon^2))^{11/2}$ \cite{AMN98}. 
In particular, we
note that our construction improves the Azar-Motwani-Naor construction
significantly in the parameters $d$ and $n$. 

It is interesting to ask if we can obtain expanding generating sets of
smaller size in deterministic polynomial time. For an upper bound, by
the Alon-Roichman theorem we know that there exist expanding generating
sets of size $O(\frac{1}{\lambda^2}\log |G|)$ for any $G\leq S_n$, which is 
bounded by $O(n\log n/\lambda^2)=\tilde{O}(n/\lambda^2)$. 
In general, given $G$, an algorithmic question is to
ask for a minimum size expanding generating set for $G$ that makes the
Cayley graph $\lambda$-spectral expander.

In this connection, it is interesting to note the following negative
result that Lubotzky and Weiss \cite{LW93} have shown about
solvable groups: Let $\{G_i\}$ be any infinite family of
finite solvable groups $\{G_i\}$ such that each $G_i$ has derived
series of length bounded by some constant $\ell$. Further, suppose
that $\Sigma_i$ is an arbitrary generating set for $G_i$ such that its
size $|\Sigma_i|\le k$ for each $i$ and some constant $k$. Then the
Cayley graphs $\Cay(G_i,\Sigma_i)$ do not form a family of expanders.
In contrast, they also exhibit an infinite family of solvable groups
in \cite{LW93} that give rise to constant-degree Cayley expanders.



\section{Combining Generating Sets for Normal subgroup and Quotient Group}\label{sec:normal-quotient}

Let $G$ be any finite group and $N$ be a normal subgroup of $G$ (i.e.\
$g^{-1}Ng=N$ for all $g\in G$). We denote this by $G\rhd N\rhd
\{1\}$. Let $A \subset N$ be an expanding generating set for $N$ and
$\Cay(N,A)$ be a $\lambda$-spectral expander.  Similarly, suppose
$B\subset G$ such that $\hat{B}=\left\{Nx~\vert~x\in B \right\}$ is an
expanding generating set for the quotient group $G/N$ and
$\Cay(G/N,\hat{B})$ is also a $\lambda$-spectral expander.  Let
$X=\{x_1,x_2,\ldots,x_k\}$ denote a set of distinct coset
representatives for the normal subgroup $N$ in $G$.  In this section
we show that $\Cay(G,A\cup B)$ is a $\frac{1+\lambda}{2}$-spectral
expander.

In order to analyze the spectral expansion of the Cayley graph
$\Cay(G,A\cup B)$ it is useful to view vectors in $\mathbb{C}^{|G|}$ as
elements of the group algebra $\mathbb{C}[G]$. The group algebra
$\mathbb{C}[G]$ consists of linear combinations $\sum_{g\in G}\alpha_g
g$ for $\alpha_g\in\mathbb{C}$. Addition in $\C[G]$ is component-wise,
and clearly $\mathbb{C}[G]$ is a $\card{G}$-dimensional vector space
over $\mathbb{C}$. The product of $\sum_{g\in G}\alpha_g g$ and
$\sum_{h\in G}\beta_h h$ is defined naturally as: $\sum_{g,h\in
  G}\alpha_g\beta_h gh$.

Let $S\subset G$ be any symmetric subset and let $M_S$ denote the
normalized adjacency matrix of the undirected Cayley graph $\Cay(G,S)$.
Now, each element $a\in G$ defines the linear map $M_a:\C[G]\to \C[G]$
by $M_a(\sum_g \alpha_gg)=\sum_g \alpha_gga$. 
Clearly, $M_S=\frac{1}{|S|}\sum_{a\in S}M_a$ and 
$M_S(\sum_g\alpha_gg)=\frac{1}{|S|}\sum_{a\in S}\sum_g\alpha_gga$.

In order to analyze the spectral expansion of $\Cay(G,A\cup B)$ we 
consider the basis $\{xn\mid x\in X, n\in N\}$ of $\C[G]$. The element
$u_N=\frac{1}{|N|}\sum_{n\in N}n$ of $\C[G]$ corresponds to the uniform
distribution supported on $N$. It has the following important
properties:

\begin{enumerate}
\item For all $a\in N$ $M_a(u_N)=u_N$ because $Na=N$ for each $a\in N$.

\item For any $b\in G$ consider the linear map $\sigma_b:\C[G]\to \C[G]$
defined by conjugation: $\sigma_b(\sum_g\alpha_gg)=\sum_g\alpha_gb^{-1}gb$.
Since $N\lhd G$ the linear map $\sigma_b$ is an automorphism of $N$.
It follows that for all $b\in G$ $\sigma_b(u_N)=u_N$.
\end{enumerate}

Now, consider the subspaces $U$ and $W$ of $\C[G]$ defined as follows:
\begin{eqnarray*}
  U= \left\{ \left(\sum_{x\in X}\alpha_x x\right) u_N \right\},~~~
  W= \left\{ \sum_{x\in X} x \left( \sum_{n\in
  N}\beta_{n,x} n \right)~\Bigl\lvert~\sum_{n}\beta_{n,x}=0,~\forall x\in X\right\}
\end{eqnarray*}
It is easy to see that $U$ and $W$ are indeed subspaces of
$\C[G]$. Furthermore, we note that every vector in $U$ is orthogonal
to every vector in $W$, i.e.\ $U\perp W$. This follows easily from the
fact that $xu_N$ is orthogonal to $x\sum_{n\in N}\beta_{n,x}n$
whenever $\sum_{n\in N}\beta_{n,x}n$ is orthogonal to $u_N$. Note that
$\sum_{n\in N}\beta_{n,x}n$ is indeed orthogonal to $u_N$ when
$\sum_{n\in N}\beta_{n,x}=0$. We claim that $\C[G]$ is a direct sum of
its subspaces $U$ and $W$.

\begin{proposition}\label{prop:v-u-w-dim}
The group algebra $\C[G]$ has a direct sum decomposition $\C[G]=U+W$.
\end{proposition}
\begin{proof}
Since $U\perp W$, it suffices to check that
$\dim(U)+\dim(W)=\card{G}$. The set $\{xu_N~\mid x\in X\}$ forms an
orthogonal basis for $U$ since for any $x\neq y\in X$, $xu_N$ is
orthogonal to $yu_N$. The cardinality of this basis is $|X|$. 
Let $z_1,\dots,z_{\card{N}-1}$ be the $|N|-1$ vectors orthogonal to
the uniform distribution $u_N$ in the eigenbasis for the Cayley graph
$\Cay(N,A)$. It is easy to see that the set $\{ xz_j~\mid x\in X, 1\le
j\le |N|-1\}$ of size $|X|(|N|-1)$ forms a basis for $W$.
\end{proof}
We will now prove the main result of this section.

\begin{lemma}  \label{lem:main-lemma}
  Let $G$ be any finite group and $N$ be a normal subgroup of $G$ and
  $\lambda<1/2$ be any constant. Suppose $A$ is an expanding generating
  set for $N$ so that $\Cay(N,A)$ is a $\lambda$-spectral
  expander. Furthermore, suppose $B\subseteq G$ such that
  $\hat{B}=\left\{Nx~\vert~x\in B \right\}$ is an
  expanding generator for the quotient group $G/N$ and $\Cay(G/N,
  \hat{B})$ is also a $\lambda$-spectral expander.  Then $A\cup B$ is
  an expanding generating set for $G$ such that $\Cay(G, A\cup B)$ is a
  $\frac{(1+\lambda)(\max{|A|,|B|})}{|A|+|B|}$-spectral expander. In
  particular, if $|A|=|B|$ then $\Cay(G, A\cup B)$ is a
  $\frac{(1+\lambda)}{2}$-spectral expander.\footnote{The sizes of $A$
    and $B$ is not a serious issue for us. Since we consider multisets
    as expanding generating sets, notice that we always ensure $|A|$
    and $|B|$ are within a factor of $2$ of each other by scaling the
    smaller multiset appropriately. Indeed, in our construction we
    can even ensure when we apply this lemma that the multisets $A$
    and $B$ are of the same cardinality which is a power of $2$.}
\end{lemma}

\begin{proof}
We will give the proof only for the case when $|A|=|B|$ (the general
case is identical).

Let $v\in \C[G]$ be any vector such that $v\perp \one$ and $M$ denote
the adjacency matrix of the Cayley graph $\Cay(G,A\cup B)$. Our goal
is to show that $\|Mv\|\le\frac{1+\lambda}{2}\|v\|$. Notice that the
adjacency matrix $M$ can be written as $\frac{1}{2}\left( M_A+M_{B}
\right)$ where $M_A=\frac{1}{\card{A}}\sum_{a\in A}M_a$ and
$M_B=\frac{1}{\card{B}}\sum_{b\in B}M_b$.\footnote{In the case when
$|A|\neq |B|$, the adjacency matrix $M$ will be
$\frac{|A|}{|A|+|B|}M_A+\frac{|B|}{|A|+|B|}M_B$.}\\

\begin{claim}\label{prop:u-w-invariant}
  For any two vectors $u\in U$ and $w\in W$, we have $M_Au\in U$, $M_Aw\in
  W$, $M_Bu\in U$, $M_Bw\in W$, i.e.\ $U$ and $W$ are invariant under the
  transformations $M_A$ and $M_B$.
\end{claim}

\begin{proof}
Consider vectors of the form $u=xu_N\in U$ and $w=x\sum_{n\in
  N}\beta_{n,x}n$, where $x\in X$ is arbitrary. By linearity, it
suffices to prove for each $a\in A$ and $b\in B$ that $M_au\in U$,
$M_bu\in U$, $M_aw\in W$, and $M_bw\in W$. Notice that
$M_au=xu_Na=xu_N=u$ since $u_Na=u_N$. Furthermore, we can write
$M_aw=x\sum_{n\in N}\beta_{n,x}na=x\sum_{n'\in N}\gamma_{n',x}n'$,
where $\gamma_{n',x}=\beta_{n,x}$ and $n'=na$. Since $\sum_{n'\in
  N}\gamma_{n',x}=\sum_{n\in N}\beta_{n,x}=0$ it follows that $M_aw\in
W$. Now, consider $M_bu=ub$. For $x\in X$ and $b\in B$ the element
$xb$ can be \emph{uniquely} written as $x_bn_{x,b}$, where $x_b\in X$
and $n_{x,b}\in N$.
  \begin{equation*}
    \begin{split}
    M_bu &= xu_Nb= xb(b^{-1}u_Nb)
    = x_bn_{x,b}\sigma_b(u_N)=x_bn_{x,b}u_N=x_bu_N\in U.
  \end{split}
  \end{equation*}
Finally, 
\begin{equation*}
    \begin{split}
    M_bw &= x(\sum_{n\in N}\beta_{n,x}n)b= xb(\sum_{n\in N}\beta_{n,x}b^{-1}nb)
    = x_bn_{x,b}\sum_{n\in N}\beta_{bnb^{-1},x}n
    = x_b\sum_{n\in N}\gamma_{n,x}n\in W.
  \end{split}
  \end{equation*}

Here, we note that $\gamma_{n,x}=\beta_{n',x}$ and
$n'=b(n^{-1}_{x,b}n)b^{-1}$. Hence $\sum_{n\in N}\gamma_{n,x}=0$, which puts
$M_bw$ in the subspace $W$ as claimed.
\end{proof}

\begin{claim}\label{prop:norm-ma-mb}
Let $u\in U$ such that $u\perp \one$ and $w\in W$. Then:
  \begin{eqnarray*}
    1.~ \|M_Au\|\le\|u\|,~~
    2.~ \|M_Bw\|\le\|w\|,~~
    3.~ \|M_Bu\|\le\lambda\|u\|,~~
    4.~ \|M_Aw\|\le\lambda\|w\|.~~
  \end{eqnarray*}
\end{claim}

\begin{proof}
Since $M_A$ is the normalized adjacency matrix of the Cayley graph
$\Cay(G,A)$ and $M_B$ is the normalized adjacency matrix of the Cayley
graph $\Cay(G,B)$, it follows that for any vectors $u$ and $w$ we
have the bounds $\|M_Au\|\le \|u\|$ and $\|M_Bw\|\le \|w\|$.

Now we prove the third part. Let $u= ( \sum_{x}\alpha_x x )u_N$ be any
vector in $U$ such that $u\perp\one$. Then $\sum_{x\in
  X}\alpha_x=0$. Now consider the vector $\hat{u}=\sum_{x\in
  X}\alpha_x Nx$ in the group algebra $\C[G/N]$. Notice that
$\hat{u}\perp \one$. Let $M_{\hat{B}}$ denote the normalized adjacency
matrix of $\Cay(G/N,\hat{B})$. Since it is a $\lambda$-spectral
expander it follows that $\|M_{\hat{B}}\hat{u}\| \le
\lambda\|\hat{u}\|$. Writing out $M_{\hat{B}}\hat{u}$ we get
$M_{\hat{B}}\hat{u}=\frac{1}{|B|}\sum_{b\in B}\sum_{x\in
  X}\alpha_xNxb=\frac{1}{|B|}\sum_{b\in B}\sum_{x\in X}\alpha_xNx_b$,
because $xb=x_bn_{x,b}$ and $Nxb=Nx_b$ (as $N$ is a normal
subgroup). Hence the norm of the vector $\frac{1}{|B|}\sum_{b\in
  B}\sum_{x\in X}\alpha_xNx_b$ is bounded by $\lambda\|\hat{u}\|$.
Equivalently, the norm of the vector $\frac{1}{|B|}\sum_{b\in
  B}\sum_{x\in X}\alpha_xx_b$ is bounded by $\lambda\|\hat{u}\|$.  On
the other hand, we have
\begin{align*}
M_Bu &= \frac{1}{\card{B}}\sum_{b}\left(\sum_{x}\alpha_x x \right)u_Nb
     = \frac{1}{\card{B}}\sum_{b}\left(\sum_{x}\alpha_x xb \right)b^{-1}u_Nb\\
     &= \frac{1}{\card{B}}\left(\sum_{b}\sum_{x}\alpha_x x_bn_{x,b}\right)u_N
     = \frac{1}{\card{B}}\left(\sum_{b}\sum_{x}\alpha_x x_b\right)u_N\\
      \end{align*}
For any vector $(\sum_{x\in X}\gamma_xx)u_N\in U$ it is easy to see
that the norm $\|(\sum_{x\in X}\gamma_xx)u_N\|=\|\sum_{x\in
  X}\gamma_xx\|\|u_N\|$. Therefore, 
\begin{eqnarray*}
\|M_Bu\|~= ~ \|\frac{1}{\card{B}}\sum_{b}\sum_{x}\alpha_x x_b\|\|u_N\|
   ~\le ~ \lambda\|\sum_{x\in X}\alpha_xx\|\|u_N\|
    ~= ~ \lambda\|u\|.
\end{eqnarray*}
    
We now show the fourth part. For each $x\in X$ it is useful to
consider the following subspaces of $\C[G]$
\[
\C[xN]=\{x\sum_{n\in N}\theta_nn\mid \theta_n\in\C\}.
\]
For any distinct $x\neq x'\in X$, since $xN\cap x'N=\emptyset$,
vectors in $\C[xN]$ have support disjoint from vectors in
$\C[x'N]$. Hence $\C[xN]\perp \C[x'N]$ which implies that the
subspaces $\C[xN], x\in X$ are pairwise mutually
orthogonal. Furthermore, the matrix $M_A$ maps $\C[xN]$ to $\C[xN]$
for each $x\in X$. 

Now, consider any vector $w=\sum_{x\in
  X}x\left(\sum_n\beta_{n,x}n\right)$ in $W$. Letting
$w_x=x\left(\sum_{n\in N}\beta_{n,x}n\right)\in\C[xN]$ for each $x\in
X$ we note that $M_Aw_x\in\C[xN]$ for each $x\in X$. Hence, by
Pythogoras theorem we have $\|w\|^2 = \sum_{x\in X}\|w_x\|^2$
and $\|M_Aw\|^2 = \sum_{x\in X}\|M_Aw_x\|^2$. Since
$M_Aw_x=xM_A\left(\sum_{n\in N}\beta_{n,x}n\right)$, it follows that
$\|M_Aw_x\| =\|M_A\left(\sum_{n\in N}\beta_{n,x}n\right)\|\le
\lambda\|\sum_{n\in N}\beta_{n,x}n\| = \lambda\|w_x\|$.

Putting it together, it follows that $\|M_Aw\|^2 \le
\lambda^2\left(\sum_{x\in X}\|w_x\|^2\right)=\lambda^2\|w\|^2$.
\end{proof}

We now complete the proof of the lemma. Consider any vector
$v\in\mathbb{C}[G]$ such that $v\perp \one$. Let $v=u+w$ where $u\in
U$ and $w\in W$. Let $\langle, \rangle$ denote the inner product
in $\C[G]$. Then we have
  \begin{align*}
    \|Mv\|^2 &= \frac{1}{4}\|(M_A+M_B)v\|^2
    = \frac{1}{4}\langle(M_A+M_B)v, (M_A+M_B)v\rangle\\
    &=\frac{1}{4}\langle M_Av,M_Av\rangle + \frac{1}{4}\langle M_Bv,M_Bv\rangle + \frac{1}{2}\langle
    M_Av,M_Bv\rangle
  \end{align*}
  We consider each of the three summands in the above expression.
  \begin{align*}
    \langle M_Av, M_Av\rangle &= \langle M_A(u+w), M_A(u+w)\rangle
    =\langle M_Au, M_Au\rangle + \langle M_Aw, M_Aw\rangle + 2\langle M_Au,
    M_Aw\rangle.
  \end{align*}
  By Claim \ref{prop:u-w-invariant} and the fact that $U\perp W$,
  $\langle M_Au, M_Aw\rangle=0$. Thus we get
  \begin{align*}
    \langle M_Av, M_Av\rangle = \langle M_Au, M_Au\rangle + \langle M_Aw,
    M_Aw\rangle
    \le \|u\|^2 + \lambda^2\|w\|^2,~\text{from Claim
    \ref{prop:norm-ma-mb}}.
  \end{align*}
  By an identical argument, Claim \ref{prop:u-w-invariant} and
  Claim \ref{prop:norm-ma-mb} imply $\langle M_Bv, M_Bv\rangle \le
  \lambda^2\|u\|^2 + \|w\|^2$.
Finally
  \begin{align*}
    \langle M_Av, M_Bv\rangle &= \langle M_A(u+w), M_B(u+w)\rangle\\
    &= \langle M_Au, M_Bu\rangle + \langle M_Aw, M_Bw\rangle + \langle
    M_Au, M_Bw\rangle + \langle M_Aw, M_Bu\rangle\\
    &= \langle M_Au, M_Bu\rangle + \langle M_Aw, M_Bw\rangle\\
    &\le \|M_Au\|\|M_Bu\| + \|M_Aw\|\|M_Bw\| ~\text{(by Cauchy-Schwarz 
inequality)}\\
    &\le \lambda\|u\|^2 + \lambda\|w\|^2,~\text{which follows from Claim
  \ref{prop:norm-ma-mb}}
  \end{align*}
  Combining all the inequalities, we get
  \begin{align*}
    \|Mv\|^2 &\le \frac{1}{4}\left( 1+2\lambda+\lambda^2 \right) \left(
    \|u\|^2 + \|w\|^2 \right)=\frac{(1+\lambda)^2}{4}\|v\|^2.
  \end{align*}
  Hence, it follows that $\|Mv\|\le\frac{1+\lambda}{2}\|v\|$.
\end{proof}

Notice that $\Cay(G,A\cup B)$ is only a $\frac{1+\lambda}{2}$-spectral
expander. We can compute another expanding generating set $S$ for $G$
from $A\cup B$, using \emph{derandomized squaring} \cite{RozenmanV05},
such that $\Cay(G,S)$ is a $\lambda$-spectral expander. We describe this
step in Appendix \ref{sec:derandomized-squaring}. As a consequence, we
obtain the following lemma which we will use repeatedly in the rest of the
paper. 
For ease of exposition, we fix $\lambda=1/4$ in the following 
lemma. 

\begin{lemma}\label{main}
  Let $G$ be a finite group and $N$ be a normal subgroup of $G$ such
  that $N=\langle A\rangle$ and $\Cay(N,A)$ is a $1/4$-spectral
  expander. Further, let $B\subseteq G$ and $\hat{B}=\{ Nx~\mid~x\in B
  \}$ such that $G/N=\langle \hat{B}\rangle$ and $\Cay(G/N, \hat{B})$
  is a $1/4$-spectral expander. Then in time
  polynomial\footnote{Though the lemma holds for any finite group $G$,
    the caveat is that the group operations in $G$ should be
    polynomial-time computable. Since we focus on permutation groups
    in this paper we will require it only for quotient groups $G=H/N$
    where $H$ and $N$ are subgroups of $S_n$.} in $|A|+|B|$, we can
  construct an expanding generating set $S$ for $G$, such that
  $|S|=O(|A|+|B|)$ and $\Cay(G,S)$ is a $1/4$-spectral expander.
\end{lemma}

\section{Normal Series and Solvable Permutation Groups}\label{sec:solvable}
In section \ref{sec:normal-quotient}, it was shown how to construct 
an expanding generating set for a group $G$ from the expanding generating 
sets of its normal subgroup
$N$ and quotient group $G/N$. In this section, we apply it to the
entire normal series for a \emph{solvable} group $G$. 
More precisely, let $G\le S_n$ such that 
$G=G_0\triangleright G_1\triangleright\dots\triangleright G_r=\{1\}$
is a \emph{normal series} for $G$. Thus $G_i$ is a normal subgroup
  of $G$ for each $i$ and hence $G_i$ is a normal subgroup of $G_j$
for each $j<i$. We give a construction of an expanding generating
set for $G$, when the expanding generating sets for the quotient
groups $G_i/G_{i+1}$ are known.

\begin{lemma}\label{normal}
Let $G\le S_n$ with normal series $\{G_i\}_{i=0}^{r}$ as above.
Further, for each $i$ let $B_i$ be a generating set for $G_i/G_{i+1}$
such that $\Cay(G_i/G_{i+1},B_i)$ is a $1/4$-spectral expander. Let
$s=\max_{i}\{|B_i|\}$. Then in deterministic time polynomial in
$n$ and $s$ we can compute a generating set $B$ for $G$ such that
$\Cay(G,B)$ is a $1/4$-spectral expander and $|B|=c^{\log r}s$ for
some constant $c>0$.
\end{lemma}

\begin{proof}
The proof is an easy application of Lemma~\ref{main}. First suppose we
have three indices $k,\ell,m$ such that $G_k\rhd G_\ell\rhd G_m$ and
$\Cay(G_k/G_\ell,S)$ and $\Cay(G_\ell/G_m,T)$ both are $1/4$-spectral
expanders. Then notice that we have the groups $G_k/G_m\rhd
G_\ell/G_m\rhd \{1\}$ and the group $\frac{G_k}{G_\ell}$ is isomorphic
to $\frac{G_k/G_m}{G_\ell/G_m}$ via a natural isomorphism. Hence
$\Cay(\frac{G_k/G_m}{G_\ell/G_m},\hat{S})$ is also a $1/4$-spectral
expander, where $\hat{S}$ is the image of $S$ under the said natural
isomorphism. Therefore, we can apply Lemma~\ref{main} by setting $G$
to $G_k/G_m$ and $N$ to $G_\ell/G_m$ to get a generating set $U$ for
$G_k/G_m$ such that $\Cay(G_k/G_m,U)$ is $1/4$-spectral and $|U|\le
c(|S|+|T|)$.

To apply this inductively to the entire normal series, assume without loss
of generality, its length to be $r=2^t$. 
Inductively assume that in the normal series 
 $G=G_0\rhd G_{2^i}\rhd G_{2\cdot2^i}\rhd G_{3\cdot2^i} \dots\triangleright
 G_r=\{1\}$,
for each quotient group $G_{j2^i}/G_{(j+1)2^i}$ we have an expanding
generating set of size $c^is$ that makes $G_{j2^i}/G_{(j+1)2^i}$
$1/4$-spectral. Now, consider the three groups $G_{(2j)2^i}\rhd
G_{(2j+1)2^i}\rhd G_{(2j+2)2^i}$ and setting $k=2j2^i$, $\ell=(2j+1)2^i$
and $m=(2j+2)2^i$ in the above argument we get expanding generating
sets for $G_{2j2^i}/G_{(2j+2)2^i}$ of size $c^{i+1}s$ that makes it
$1/4$-spectral. The lemma follows by induction.
\end{proof}

\subsection{Solvable permutation groups}
Now we apply the above lemma to solvable permutation groups. Let $G$
be any finite solvable group. The \emph{derived series} for $G$ is the
following chain of subgroups of $G$:
 $G=G_0\triangleright G_1\triangleright\dots\triangleright G_k=\{1\}$
where, for each $i$, $G_{i+1}$ is the \emph{commutator subgroup} of
$G_i$. That is $G_{i+1}$ is the normal subgroup of $G_i$ generated by
all elements of the form $xyx^{-1}y^{-1}$ for $x,y\in G_i$. It turns
out that $G_{i+1}$ is the minimal normal subgroup of $G_i$ such that
$G_i/G_{i+1}$ is abelian. Furthermore, the derived series is also a
\emph{normal series}. That means each $G_i$ is in fact a normal subgroup of
$G$ itself. It also implies that $G_i$ is a normal subgroup of $G_j$
for each $j<i$.

Our algorithm will crucially exploit a property of the derived series
of solvable groups $G\le S_n$: By a theorem of Dixon
\cite{Dixon68}, the length $k$ of the derived series
of a solvable subgroup of $S_n$ is bounded by $5\log_{3}n$. 
Thus, we get the following result as a direct application of Lemma \ref{normal}:

\begin{lemma}\label{solv}
Suppose $G\le S_n$ is a solvable group with derived series
 $G=G_0\triangleright G_1\triangleright\dots\triangleright G_k=\{1\}$
such that for each $i$ we have an expanding generating set $B_i$ for
the abelian quotient group $G_i/G_{i+1}$ such that
$\Cay(G_i/G_{i+1},B_i)$ is a $1/4$-spectral expander. Let
$s=\max_{i}\{|B_i|\}$. Then in deterministic time polynomial in $n$
and $s$ we can compute a generating set $B$ for $G$ such that
$\Cay(G,B)$ is a $1/4$-spectral expander and $|B|=2^{O(\log k)}s=(\log
n)^{O(1)}s$.
\end{lemma}


Given a solvable permutation group $G\le S_n$ by a generating set the
polynomial-time algorithm for computing an expanding generating set
will proceed as follows: in deterministic polynomial time, we first
compute \cite{Luk93} generating sets for each subgroup $\{G_i\}_{1\leq
 i\leq k}$ in the derived series for $G$. In order to apply the
above lemma it suffices to compute an expanding generating set $B_i$
for $G_i/G_{i+1}$ such that $\Cay(G_i/G_{i+1},B_i)$ is $1/4$-spectral.
We deal with this problem in the next section.

\section{Abelian Quotient Groups}\label{sec:abelian-group}
In Section \ref{sec:solvable}, we have seen how to construct
an expanding generating set for a solvable group $G$, from
expanding generating sets for the quotient groups $G_i/G_{i+1}$
in the normal series for $G$.
We are now left with the problem of computing
expanding generating sets for the abelian quotient groups
$G_i/G_{i+1}$. We state a couple of easy lemmas that will allow us to
further simplify the problem. We defer the proofs of these lemmas to 
Appendix \ref{sec:abel-perm}.

\begin{lemma}\label{abel1}
Let $H$ and $N$ be subgroups of $S_n$ such that $N$ is a normal
subgroup of $H$ and $H/N$ is abelian. Let $p_1<p_2<\ldots < p_k$ be
the set of all primes bounded by $n$ and $e=\lceil \log n\rceil$.
Then, there is an onto homomorphism $\phi$ from the product group
$\Z_{p_1^e}^n\times\Z_{p_2^e}^n\times\cdots\times\Z_{p_k^e}^n$ to
the abelian quotient group $H/N$.
\end{lemma}

Suppose $H_1$ and $H_2$ are two finite groups such that $\phi:H_1\to
H_2$ is an onto homomorphism. In the next lemma we show that the
$\phi$-image of an expanding generating set for $H_1$, is an expanding
generating set for $H_2$.
 
\begin{lemma}\label{abel2}\label{lem:hom-expand}\label{lem:onto-hom-exp}
Suppose $H_1$ and $H_2$ are two finite groups such that $\phi:H_1\to
H_2$ is an onto homomorphism. Furthermore, suppose $\Cay(H_1,S)$ is a
$\lambda$-spectral expander. Then $\Cay(H_2,\phi(S))$ is also a
$\lambda$-spectral expander.
\end{lemma}

Now, suppose $H,N\le S_n$ are groups given by their generating sets, where
$N\lhd H$ and $H/N$ is abelian. By Lemmas \ref{abel1} and \ref{abel2},
it suffices to describe a polynomial (in $n$) time algorithm for
computing an expanding generating set of size $\tilde{O}(n^2)$ for the
product group
$\Z_{p_1^e}^n\times\Z_{p_2^e}^n\times\cdots\times\Z_{p_k^e}^n$ such that 
the second largest eigenvalue of the corresponding Cayley graph is 
bounded by $1/4$. In the following section, we solve this problem. 

 
\subsection{Expanding generating set for the product group}
\label{sec:epsilon-base}

In this section, we give a deterministic polynomial (in $n$) time
construction of an $\tilde{O}(n^2)$ size expanding generating set for
the product group
${\Z_{p_1^{e}}^n}\times{\Z_{p_2^{e}}^n}\times\ldots\times{\Z_{p_k^{e}}^n}$ 
such that the second largest eigenvalue of the corresponding 
Cayley graph is bounded by $1/4$.

Consider the following \emph{normal series} for this product group
given by the subgroups
$K_i={\Z_{p_1^{e-i}}^n}\times{\Z_{p_2^{e-i}}^n}\times\ldots
\times{\Z_{p_k^{e-i}}^n}$ for $0\leq i\leq e$. Clearly,
$K_0\triangleright K_1\triangleright\dots\triangleright K_e=\{1\}.$
This is obviously a normal series since
$K_0={\Z_{p_1^{e}}^n}\times{\Z_{p_2^{e}}^n}\times\ldots\times{\Z_{p_k^{e}}^n}$
is abelian. Furthermore,
$K_i/K_{i+1}=\Z_{p_1}^n\times\Z_{p_2}^n\times\ldots\times\Z_{p_k}^n$.

Since the length of this series is $e=\lceil\log n\rceil$ we can apply
Lemma~\ref{normal} to construct an expanding generating set of size
$\tilde{O}(n^2)$ for $K_0$ in polynomial time assuming that we can
compute an expanding generating set of size $\tilde{O}(n^2)$ for
$\Z_{p_1}^n\times\Z_{p_2}^n\times\ldots\times\Z_{p_k}^n$ in
deterministic polynomial time.
Thus, it suffices to efficiently compute an $\tilde{O}(n^2)$-size
expanding generating set for the product group
$\Z_{p_1}^n\times\Z_{p_2}^n\times\ldots\times\Z_{p_k}^n$.

In \cite{AIKPS90}, Ajtai et al, using some number theory, gave a
deterministic polynomial time expanding generating set construction for
the cyclic group $\Z_t$, where $t$ is given in \emph{binary}. 

\begin{theorem}[\cite{AIKPS90}]\label{thm:aikps}
Let $t$ be a positive integer given in binary as an input. Then there
is a deterministic polynomial-time (i.e. in $\poly(\log t)$ time)
algorithm that computes an expanding generating set $T$ for $\Z_t$ of
size $O(\log^*t\log t)$, where $\log^*t$ is the least positive integer $k$ 
such that a tower of $k$ $2$'s bounds $t$. Furthermore, $\Cay(\Z_t,T)$
is $\lambda$-spectral for any constant $\lambda$.
\end{theorem}

Now, consider the group $\Z_{p_1 p_2 \ldots p_k}$. Since $p_1 p_2 \ldots
p_k$ can be represented by $O(n\log n)$ bits in binary, we apply the above
theorem (with $\lambda=1/4$) to compute an expanding generating set of size
$\tilde{O}(n)$ for $\Z_{p_1 p_2 \ldots p_k}$ in $\poly(n)$ time. Let
$m=O(\log n)$ be a positive integer to be fixed in the analysis later.
Consider the product group $M_0=\Z_{p_1}^m\times
\Z_{p_2}^m\times\dots\Z_{p_k}^m$ and for $1\leq i\leq m$ let
$M_i=\Z_{p_1}^{m-i}\times \Z_{p_2}^{m-i}\times \ldots \times
\Z_{p_k}^{m-i}$. Clearly, the groups $M_i$ form a \emph{normal} series for
$M_0$:
$M_0\rhd M_1\rhd\cdots \rhd M_m=\{1\},$
and the quotient groups are
$M_i/M_{i+1}=\Z_{p_1}\times\Z_{p_2}\times\ldots\times\Z_{p_k}=\Z_{p_1p_2\ldots
  p_k}$. 
Now we compute (in $\poly(n)$ time) an expanding
generating set for $\Z_{p_1p_2\cdots p_k}$ of size $\tilde{O}(n)$ using
Theorem~\ref{thm:aikps}. Then, we apply Lemma~\ref{normal} to the above
normal series and compute an expanding generating set of size 
$\tilde{O}(n)$ for the product
group $M_0$ in polynomial time. The corresponding Cayley graph will be 
a $1/4$-spectral expander.
Now we are ready to describe the expanding generating set construction
for ${\Z_{p_1}^n}\times{\Z_{p_2}^n}\times\ldots\times{\Z_{p_k}^n}$.

\subsubsection{The final construction}\label{sec:finalconstruction} 

For $1\le i\le k$ let $m_i$ be the least positive integer such that
$p_i^{m_i}>cn$ (where $c$ is a suitably large constant). Thus,
$p_i^{m_i}\leq cn^2$ for each $i$. For each $i$, $\F_{p_i^{m_i}}$ be
the finite field of $p_i^{m_i}$ elements which can be
deterministically constructed in polynomial time since it is
polynomial sized.  Clearly, there is an onto homomorphism $\psi$ from
the group $\Z_{p_1}^m\times \Z_{p_2}^m\times \ldots\times \Z_{p_k}^m$
to the additive group of $\F_{p_1^{m_1}}\times \F_{p_2^{m_2}}\times
\ldots\times\F_{p_k^{m_k}}$.  Thus, if $S$ is the expanding generating
set of size $\tilde{O}(n)$ constructed above for $\Z_{p_1}^m\times
\Z_{p_2}^m\times \ldots\times \Z_{p_k}^m$, it follows from Lemma
\ref{abel2} that $\psi(S)$ is
an expanding generator multiset of size $\tilde{O}(n)$ for the
additive group $\F_{p_1^{m_1}}\times \F_{p_2^{m_2}}\times
\ldots\times\F_{p_k^{m_k}}$. Define $T\subset \F_{p_1^{m_1}}\times
\F_{p_2^{m_2}}\times \ldots\times\F_{p_k^{m_k}}$ to be any (say, the
lexicographically first) set of $cn$ many $k$-tuples such that any
two tuples $(x_1,x_2,\ldots,x_k)$ and $(x'_1,x'_2,\ldots,x'_k)$ in $T$
are distinct in all coordinates. Thus \ $x_j\neq x'_j$ for all $j\in[k]$.
It is obvious that we can construct $T$ by picking the first $cn$ such
tuples in lexicographic order. 

Now we will define the expanding generating set $R$.  Let
$x=(x_1,x_2,\ldots,x_k)\in T$ and $y=(y_1,y_2,\ldots,y_k)\in \psi(S)$.
Define $v_i=(y_i, \langle x_i,y_i\rangle,\langle x_{i}^2,y_i\rangle,
\ldots, \langle x_i^{n-1},y_i\rangle)$ where $x_i^j\in \F_{p_i^{m_i}}$
and $\langle x_i^j,y_i\rangle$ is the inner product modulo $p_i$ of
the elements $x_i^j$ and $y_i$ seen as $p_i$-tuples in
$\Z_{p_i}^{m_i}$.
Hence, $v_i$ is an $n$-tuple and
$v_i\in\Z_{p_i}^n$.  Now define
$R=\{(v_1,v_2,\ldots,v_k)~|~x\in T,y\in \psi(S)\}.$
Notice that $|R|=\tilde{O}(n^2)$. 
\begin{claim}\label{clm:final-claim}
$R$ is an expanding
generating set for the product group $\Z_{p_1}^n\times \Z_{p_2}^n\times
\ldots \times \Z_{p_k}^n$. 
\end{claim}

\begin{proof}
Let $(\chi_1,\chi_2,\ldots,\chi_k)$ be a
nontrivial character of the product group $\Z_{p_1}^n\times
\Z_{p_2}^n\times \ldots \times \Z_{p_k}^n$, i.e.  there is at least
one $j$ such that $\chi_j$ is nontrivial. Let $\omega_i$ be a
primitive $p_i^{th}$ root of unity. Recall that, since $\chi_i$ is a
character there is a corresponding vector $\beta_i\in\Z_{p_i}^n$,
i.e. $\chi_i : \Z_{p_i}^n\rightarrow \C$ and
$\chi_i(u)=\omega_i^{\langle \beta_i,u\rangle}$ for $u\in \Z_{p_i}^n$
and the inner product in the exponent is a modulo $p_i$ inner
product. The character $\chi_i$ is nontrivial if and only if $\beta_i$
is a nonzero element of $\Z_{p_i}^n$.

The characters $(\chi_1,\chi_2,\ldots,\chi_k)$ of the abelian
group $\Z_{p_1}^n\times \Z_{p_2}^n\times \ldots \times \Z_{p_k}^n$ are
also the eigenvectors for the adjacency matrix of the Cayley graph of the group
with any generating set.
Thus, in order to prove that $R$ is an expanding generating
set for $\Z_{p_1}^n\times \Z_{p_2}^n\times \ldots \times \Z_{p_k}^n$,
it is enough to bound the following exponential sum estimate for the
nontrivial characters $(\chi_1,\chi_2,\ldots,\chi_k)$ since that
directly bounds the second largest eigenvalue in absolute value.

\begin{eqnarray*}
\left|\E_{x\in T,y\in \psi(S)} 
[\chi_1(v_1)\chi_2(v_2)\ldots \chi_k(v)]\right|
&=& \left|\E_{x\in T, y\in \psi(S)} [\omega_1^{\langle \beta_1,v_1\rangle}\ldots 
\omega_k^{\langle \beta_k,v_k\rangle}]\right|\\
&=&\left|\E_{x\in T,y\in \psi(S)} [\omega_1^{\langle
q_1(x_1),y_1\rangle}\ldots \omega_k^{\langle q_k(x_k),y_k\rangle}]\right|\\
&\leq& \E_{x\in T}\left|\E_{y\in \psi(S)}[\omega_1^{\langle q_1(x_1),y_1\rangle}\ldots
\omega_k^{\langle q_k(x_k),y_k\rangle}]\right|, 
\end{eqnarray*}
where
$q_i(x)=\sum_{\ell=0}^{n-1}\beta_{i,\ell}x^{\ell}\in\F_{p_i}[x]$ for
$\beta_i=(\beta_{i,1},\beta_{i,2},\ldots,\beta_{i,n})$.  Since the
character is nontrivial, suppose $\beta_j\neq 0$, then $q_j$ is
a nonzero polynomial of degree at most $n-1$. Hence the probability
that $q_j(x_j)=0$, when $x$ is picked from $T$ is bounded by
$\frac{n}{cn}$.  
On the other hand, when $q_j(x_j)\neq 0$ the tuple
$(q_1(x_1),\ldots,q_k(x_k))$ defines a nontrivial character of the
group $\Z_{p_1}^m\times \ldots\times\Z_{p_k}^m$. Since $S$ is an
expanding generating set for the abelian group
$\Z_{p_1}^m\times\ldots\times \Z_{p_k}^m$, the character defined by
$(q_1(x_1),\ldots,q_k(x_k))$ is also an eigenvector for
$\Z_{p_1}^m\times\ldots\times \Z_{p_k}^m$, in particular
w.r.t.\ generating set $S$. Hence, we have that $\left|\E_{y\in
  S}[\omega_1^{\langle q_1(x_1),y_1\rangle}\ldots\omega_k^{\langle
    q_k(x_k),y_k\rangle}]\right|\leq \varepsilon$, where the parameter
$\varepsilon$ can be fixed to an arbitrary small constant by Theorem
\ref{thm:aikps}. Hence the above estimate is bounded by
$\frac{n}{cn}+\varepsilon=\frac{1}{c}+\varepsilon$ which can be made 
$\leq 1/4$ by choosing $c$ and $\epsilon$ suitably. 
\end{proof}

To summarize, Claim \ref{clm:final-claim} along with 
Lemmas \ref{abel1} and \ref{abel2}
directly yields the following theorem.

\begin{theorem}\label{abel-thm}
In deterministic polynomial (in $n$) time we can
construct an expanding generating set of size $\tilde{O}(n^2)$ for the
product group $\Z_{p_1}^n\times\cdots\times \Z_{p_k}^n$ 
(where for each $i$, $p_i$ is a prime number $\leq n$) that makes it
a $1/4$-spectral expander. Consequently, if $H$ and $N$ are subgroups of
$S_n$ given by generating sets and $H/N$ is abelian then in
deterministic polynomial time we can compute an expanding generating
set of size $\tilde{O}(n^2)$ for $H/N$ that makes it a 
$1/4$-spectral expander.
\end{theorem}

Finally, we state the main theorem which follows directly from
the above theorem and Lemma~\ref{solv}.

\begin{theorem}\label{final}
Let $G\le S_n$ be a solvable permutation group given by a generating
set. Then in deterministic polynomial
time we can compute an expanding generating set $S$ of size
$\tilde{O}(n^2)$ such that the Cayley graph $\Cay(G,S)$ is a
$1/4$-spectral expander.
\end{theorem} 

On a related note, in the case of general permutation groups we have the
following theorem about computing expanding generating sets.
\begin{theorem}\label{thm:gen-perm-group}
Given $G\le S_n$ by a generating set $S'$ and $\lambda>0$, we can
deterministically compute (in time $\poly(n,|S'|)$) an expanding
generating set $T$ for $G$ such that $\Cay(G,T)$ is a
$\lambda$-spectral expander and $|T|=O(n^{16q+10}
\left(\frac{1}{\lambda}\right)^{32q})$ (where $q$ is the constant in
Lemma~\ref{lem:explicit-aux-exp-family}).
\end{theorem}

For a proof-sketch of the above theorem, refer Appendix
\ref{sec:permutation-group}. Using the same method as in Appendix
\ref{sec:permutation-group} we can observe that
for any $\lambda$, the size of the expanding generating set 
$S$ given by Theorem \ref{final} is $\tilde{O}(n^2)(1/\lambda)^{32q}$ 
when $G$ is a solvable subgroup of $S_n$.

\section{Small Bias Spaces for $\Z_d^n$}\label{sec:discussion}
In Section \ref{sec:abelian-group}, we constructed expanding 
generating sets for abelian groups. We note that this also gives a new
construction of $\varepsilon$-bias spaces for $\Z_d^n$,
which we describe in this section.

In \cite{AMN98} Azar, Motwani, and Naor first considered the
construction of $\varepsilon$-bias spaces for abelian groups,
specifically for the group $\Z_d^n$. For arbitrary $d$ and any
$\varepsilon>0$ they construct $\varepsilon$-bias spaces of size
$O((d+n^2/\varepsilon^2)^C)$, where $C$ is the constant in Linnik's
Theorem. The construction involves finding a suitable prime (or prime
power) promised by Linnik's theorem which can take time up to
$O((d+n^2)^C)$. The current best known bound for $C$ is $\leq 11/2$
(and assuming ERH it is $2$). Their construction yields a
polynomial-size $\varepsilon$-bias space for $d=n^{O(1)}$. 

It is interesting to compare this result of \cite{AMN98} with our results.  
Let $d$ be any positive integer with prime
factorization $p_1^{e_1}p_2^{e_2}\cdots p_k^{e_k}$. So each
$p_i$ is $O(\log d)$ bit sized and each $e_i$ is bounded by $O(\log d)$.  
Given $d$ as input in unary, we can efficiently find the prime factorization
of $d$. Using the result of Wigderson and Xiao \cite{WigdersonX08}, we compute 
an $O(\log d)$ size expanding generating set for 
$\Z_{p_1 p_2 \ldots p_k}$ in deterministic 
time polynomial in $d$. Then we construct an expanding generating set of size 
$O(\poly(\log n)\log d)$ for 
$\Z_{p_1}^m\times \Z_{p_2}^m\times\ldots\times\Z_{p_k}^m$ for $m=O(\log n)$ 
using the method described in Section \ref{sec:epsilon-base}. It then follows 
from Section \ref{sec:finalconstruction} that we can construct an 
$O(n \poly(\log n)\log d)$ size expanding generating set for 
$\Z_{p_1}^n\times\Z_{p_2}^n\times\ldots\times\Z_{p_k}^n$ in deterministic 
polynomial time. Finally, from Section \ref{sec:epsilon-base}, it follows 
that we can construct an $O(n\poly(\log n,\log d))$ size expanding 
generating set for $\Z_d^n$ (which is isomorphic to
$\Z_{p_1^{e_1}}^n\times\dots\Z_{p_k^{e_k}}^n$) since each $e_i$ is bounded by $\log d$. 
Now for any arbitrary $\varepsilon >0$, the explicit dependence of 
$\varepsilon$ in the size of the generating set is $(1/\epsilon)^{32q}$. 
We obtain it by applying the technique described in 
Section \ref{sec:permutation-group}.  
We summarize the discussion in the following theorem. 

\begin{theorem}\label{thm:improved-AMN}
Let $d,n$ be any positive integers (given in unary) and $\varepsilon > 0$.
Then, in deterministic $\poly(n,d,\frac{1}{\varepsilon})$ time, we can
construct an $O(n\poly(\log n,\log d))(1/\varepsilon)^{32q}$ size
$\varepsilon$-bias space for $\Z_d^n$.    
\end{theorem}


\section{Open Problems}\label{sec:open}
Alon-Roichman theorem guarantees the existence of $O(n\log n)$ size expanding 
generating sets for permutation groups $G\leq S_n$. In this paper, we construct 
$\tilde{O}(n^2)$ size expanding generating sets for solvable groups. 
For an arbitrary permutation group, our bound is far from optimal.   
Our construction of $\varepsilon$-bias space for $\Z_d^n$ improves upon the 
construction of \cite{AMN98} in terms of $d$ and $n$ significantly. However, it
is worse in terms of the parameter $\varepsilon$.
Improving the above bounds remains a challenging open problem.\\

\noindent\textbf{Acknowledgements.}~~We thank Shachar Lovett for
pointing out to us the result of Ajtai et al \cite{AIKPS90}.  We also
thank Avi Wigderson for his comments and suggestions.

\bibliographystyle{amsalpha}

\newpage

\begin{center}
\appendix{\huge\bf{Appendix}}
\end{center}

\section{Derandomized Squaring}\label{sec:derandomized-squaring}
We recall a
result in \cite[Observation 4.3,Theorem 4.4]{RozenmanV05} about
derandomized squaring applied to Cayley graphs in some
detail.

\begin{theorem} [\cite{RozenmanV05}]
  Let $G$ be a finite group and $U$ be an expanding generating set such
  that $\Cay(G,U)$ is a $\lambda'$-spectral expander and $H$ be a
  consistently labeled $d$-regular graph with vertex set
  $\{1,2,\ldots,|U|\}$ for a constant $d$ such that $H$ is a
  $\mu$-spectral expander. Then $\Cay(G,U)\circledS H$ is a
  \emph{directed} Cayley graph for the same group $G$ and with
  generating set $S=\{ u_iu_j~\vert~(i,j)\in E(H) \}$. Furthermore, if
  $A$ is the normalized adjacency matrix for $\Cay(G,U)\circledS H$
  then for any vector $v\in \C^{|G|}$ such that $v\perp \one$:
\[
\|Av\|\le (\lambda'^2+\mu)\|v\|.
\]
\end{theorem}

Observe that in the definition of the directed Cayley graph
$\Cay(G,U)\circledS H$ (in the statement above) there is an
identification of the vertex set $\{1,2,\ldots,|U|\}$ of $H$ with the
generator \emph{multiset} $U$ indexed as
$U=\{u_1,u_2,\ldots,u_{|U|}\}$.

Alternatively, we can also identify the vertex set
$\{1,2,\ldots,|U|\}$ of $H$ with the generator \emph{multiset} $U$
indexed as $U=\{u^{-1}_1,u^{-1}_2,\ldots,u^{-1}_{|U|}\}$, since $U$ is
closed under inverses and, as a multiset, we assume for each $u\in U$
both $u$ and $u^{-1}$ occur with the same multiplicity. Let us denote this
directed Cayley graph by $\Cay(G,U^{-1})\circledS H$. Clearly, by the
above result of \cite{RozenmanV05} the graph $\Cay(G,U^{-1})\circledS
H$ also has the same expansion property. I.e.\ if $A'$ denotes its
normalized adjacency matrix for $\Cay(G,U^{-1})\circledS H$ then for
any vector $v\in \C^{|G|}$ such that $v\perp \one$:
\[
\|A'v\|\le (\lambda'^2+\mu)\|v\|.
\]

We summarize the above discussion in the following lemma.

\begin{lemma} \label{lem:derand-squaring-rv} Let $G$ be a finite group
  and $U$ be a generator multiset for $G$ such that for each $u\in U$
  both $u$ and $u^{-1}$ occur with the same multiplicity (i.e.\ $U$ is
  symmetric and preserves multiplicities). Suppose $\Cay(G,U)$ is a
  $\lambda'$-spectral expander. Let $H$ be a consistently labeled
  $d$-regular graph with vertex set $\{1,2,\ldots,|U|\}$ for a
  constant $d$ such that $H$ is a $\mu$-spectral expander. Then
  $\Cay(G,S)$ is an \emph{undirected} Cayley graph for the same group
  $G$ and with generating set $S=\{ u_iu_j~\mid~(i,j)\in E(H)\}\cup
  \{u_i^{-1}u^{-1}_j\mid~(i,j)\in E(H)\}$. Furthermore, $\Cay(G,S)$ is
  a $(\lambda'^2+\mu)$-spectral expander of degree $2d|U|$.
\end{lemma}

We can, for instance, use the graphs given by the following lemma for
$H$ in the above construction.

\begin{lemma}[\cite{RozenmanV05}]
\label{lem:explicit-aux-exp-family}
  For some constant $Q=4^q$, there exists a sequence of consistently
  labelled $Q$-regular graphs on $Q^m$ vertices whose second largest
  eigenvalue is bounded by $1/100$ such that given a vertex $v\in [Q^m]$
  and an edge label $x\in[Q]$, we can compute the $x^{th}$ neighbour of $v$
  in time polynomial in $m$.
\end{lemma}

Suppose $\Cay(G,U)$ is a $3/4$-spectral expander and we take $H$ given by
the above lemma for derandomized squaring, then it is easy to see that with
a constant number of squaring operations we will obtain a generating set
$S$ for $G$ such that $|S|=O(|U|)$ and $\Cay(G,S)$ is a $1/4$-spectral
expander.  Putting this together with Lemma~\ref{lem:main-lemma} we obtain
Lemma \ref{main}.

\section{Proof of Lemma \ref{abel1} and Lemma \ref{abel2}}
\label{sec:abel-perm}

\begin{proof}[Proof of Lemma \ref{abel1}]
Since $H$ is a subgroup of $S_n$ it has a generating set of size
at most $n$ \cite{Jer82}. 
Let $\{x_1,x_2,\ldots,x_n\}$ be a generator (multi)set
for $H$. Each permutation $x_i$ can be written as a product of
disjoint cycles and the order, $r_i$, of $x_i$ is the lcm of
the lengths of these disjoint cycles. Thus we can write for each $i$
\[
r_i=p_1^{e_{i1}}p_2^{e_{i2}}\ldots p_k^{e_{ik}},
\]
where the key point to note is that $p_j^{e_{ij}}\le n$ for each $i$ and
$j$ because $r_i$ is the lcm of the disjoint cycles of permutation $x_i$.
Clearly, $e_{ij}\le e=\lceil \log n\rceil$.

Now, define the elements $y_{ij}=x_i^{{r_i}/{p_j^{e_{ij}}}}$. Notice that
the order, $o(y_{ij})$, of $y_{ij}$ is $p_j^{e_{ij}}$.

Let $(a_{11},\ldots,a_{n1},\ldots,a_{1k},\ldots,a_{nk})$ be an element
of the product group
$\Z_{p_1^e}^n\times\Z_{p_2^e}^n\times\cdots\times\Z_{p_k^e}^n$, where
for each $i$ we have $(a_{1i},\ldots,a_{ni})\in \Z_{p_i^e}^n$. 
Now define the mapping $\phi$ as
  \[
  \phi(a_{11},\dots,a_{n1},\dots,a_{1k},\dots,a_{nk}) =
  N(\prod_{j=1}^k\prod_{i=1}^ny_{ij}^{a_{ij}}).
\]
Since $H/N$ is abelian, it is easy to see that $\phi$ is a
homomorphism. To see that $\phi$ is onto, consider $N x_1^{f_1}\ldots
x_{\ell}^{f_{\ell}}\in H/N$. Clearly, the cyclic subgroup generated by
$x_i$ is the direct product of its $p_j$-Sylow subgroups generated by
$y_{ij}$ for $1\leq j\leq k$. Hence $x_i^{f_i}=y_{i1}^{a_{i1}}\ldots
y_{ik}^{a_{ik}}$ for some $(a_{i1},\ldots, a_{ik})\in
\Z_{p_1^{e_{i1}}}\times \ldots\times\Z_{p_k^{e_{ik}}}$. This vector
$(a_{11},\ldots,a_{nk})$ is a pre-image of $N x_1^{f_1}\ldots
  x_{\ell}^{f_{\ell}}$, implying that $\phi$ is onto.
\end{proof}

\begin{proof}[Proof of Lemma \ref{abel2}]
Let $N=Ker(\phi)$ be the kernel of the onto homomorphism $\phi$. Then
$H_1/N$ is isomorphic to $H_2$ and the lemma is equivalent to the
claim that $\Cay(H_1/N,\hat{S})$ is a $\lambda$-spectral expander,
where $\hat{S}=\{Ns\mid s\in S\}$ is the corresponding generating set
for $H_1/N$. We can check by a direct calculation that all the eigenvalues
of the normalized adjacency matrix of $\Cay(H_1/N,\hat{S})$ are also
eigenvalues of $\Cay(H_1,S)$. This claim also follows from the well-known
results in the ``expanders monograph'' \cite[Lemma 11.15,Proposition
  11.17]{HLW06}. In order to apply these results, we note that $H_1$
naturally defines a permutation action on the quotient group $H_1/N$
by $h:Nx\mapsto Nxh$ for each $h\in H_1$ and $Nx\in H_1/N$. Then the
Cayley graph $\Cay(H_1/N,\hat{S})$ is just the {\em Schreier graph} for this
action and the generating set $S$ of $H_1$. Moreover, by \cite[Proposition
  11.17]{HLW06}, all the eigenvalues of $\Cay(H_1/N,\hat{S})$
are eigenvalues of $\Cay(H_1,S)$ and the lemma follows.
\end{proof}

\section{Expanding Generator Set for any Permutation Group}
\label{sec:permutation-group}

In this section, we give a proof-sketch of Theorem
\ref{thm:gen-perm-group}.
We require the following result on expansion of
vertex-transitive graphs; recall that a graph $X$ is said to be
\emph{vertex transitive} if its automorphism group $Aut(X)$ acts
transitively on its vertex set.

\begin{theorem}{\rm\cite{Bab91}}\label{thm:bab91}
For any vertex-transitive undirected graph of degree $d$ and diameter
$\Delta$ the second largest eigenvalue of its normalized adjacency
matrix is bounded in absolute value by $1-\frac{1}{16.5 d\Delta^2}$.
\end{theorem}

We note the well-known fact that an undirected Cayley graph
$\Cay(G,S)$ is vertex transitive, given any generating set $S$ for the
group $G$. In particular, if $G\le S_n$ we know by the Schreier-Sims
algorithm \cite{Luk93} that in deterministic polynomial time we can
compute a \emph{strong} generating set $S'$ for $G$, where $|S'|\le
n^2$. In particular, $S'$ has the property that every element of $G$
is expressible as a product of $n$ elements of $S'$. As a consequence,
the diameter of the Cayley graph $\Cay(G,S')$ is bounded by
$n$. Hence by Theorem \ref{thm:bab91}, the second largest eigenvalue
of $\Cay(G,S')$ is bounded by $1-\frac{1}{16.5 n^4}$. Now we apply
derandomized squaring from \cite{RozenmanV05} to get a spectral gap
$(1-\lambda)$ for any $\lambda>0$. In particular, we use the 
following theorem from \cite{RozenmanV05}. 

\begin{theorem}\cite[Theorem 4.4]{RozenmanV05}\label{thm:rv05main}
If $X$ is a consistently labelled $K$-regular graph on $N$ vertices that is
a $\lambda$-spectral expander and $G$ is a 
$D$-regular graph on $K$ vertices that is a $\mu$-spectral expander, then $X\circledS G$ is an 
$K D$-regular graph on $N$ vertices with spectral expansion
$f(\lambda,\mu)$, where 
$f(\lambda, \mu) = 1 - (1 - \lambda^2 )(1 - \mu)$
The function $f$ is monotone increasing in $\lambda$ and $\mu$, 
and satisfies the following conditions: 
$f(\lambda, \mu) \leq \lambda^2 + \mu$,
and $1-f(1-\gamma, 1/100)\geq (3/2) \gamma$ , when $\gamma < 1/4$.
\end{theorem}

We apply the above lemma repeatedly for at most $8\log n$ times
to get a generating set $T$ for $G$ such that the Cayley graph $\Cay(G,T)$ has a
spectral gap of at least $1/4$. Further, by Lemma
\ref{lem:derand-squaring-rv}, the size of $T$ is $O(n^{16q + 10})$,
assuming that we use the expander graphs given by
Lemma~\ref{lem:explicit-aux-exp-family} for derandomized squaring.

We cannot use a constant-degree expander to increase the spectral gap
beyond a constant. For $1-\epsilon > 1/4$, we will apply the
derandomized squaring using a non-constant degree expander as
described in \cite[Section 5]{RozenmanV05}. By the analysis of
\cite{RozenmanV05}, if we apply derandomized squaring $m$ times with a
suitable non-constant degree expander then the second largest
eigenvalue (in absolute value) will be bounded by $(7/8)^{2^{m}}$. In
order to bound this by $\epsilon$ we can set $m=3 +
\log\log\frac{1}{\epsilon}$. Also, for the $i^{th}$ derandomized
squaring step the degree of the auxiliary expander graph
turns out to be $4^{q2^i}$, $1\le i\le m$. Hence the overall degree of
the final Cayley graph will become $O(n^{16q+10}4^{q(2^{m+1}-1)})$.  Then by
Lemma \ref{lem:derand-squaring-rv}, the size of the generating set
will be $|T|=O(n^{16q+10}\left(\frac{1}{\lambda}\right)^{32q})$.  
This completes the proof-sketch of Theorem \ref{thm:gen-perm-group}.


\end{document}